\documentclass{ifacconf}

\usepackage{graphicx}      
\usepackage{natbib}        
\usepackage{color}
\usepackage[active]{srcltx}
\usepackage{epsfig}
\usepackage{amsmath}
\usepackage[flalign]{mathtools}
\usepackage{amssymb}
\usepackage{framed}
\usepackage{psfrag}
\usepackage{graphicx}
\usepackage{txfonts}
\usepackage{tikz}
\usetikzlibrary{arrows}
\usepackage{ marvosym }
\usepackage{epstopdf}

\usepackage{color}
\usepackage{tikz}
\usetikzlibrary{arrows}

\newtheorem{dfn}{Definition}
\newtheorem{proof}{Proof}
\newtheorem{example}[thm]{Example}

\begin{document}

\newcommand\independent{\protect\mathpalette{\protect\independenT}{\perp}}\def\independenT#1#2{\mathrel{\rlap{$#1#2$}\mkern2mu{#1#2}}}

\newcommand{\ComplexField}{\mathbf{C}}
\newcommand{\NaturalNumbers}{\mathbf{N}}

\newcommand{\eps}{\varepsilon}
\newcommand{\Eps}{\mathcal{E}}

\newcommand{\w}{\omega}
\newcommand{\C}{\mathbf{C}}
\newcommand{\Z}{\mathbb{Z}}
\newcommand{\I}{\mathcal{I}}
\newcommand{\N}{\mathbb{N}}
\newcommand{\R}{\mathbb{R}}

\newcommand{\tfspan}{\mathrm{tfspan}}
\newcommand{\ctfspan}{\mathrm{ctfspan}}

\newcommand{\ZF}{~{}^{0}\mathcal{F}}

\newcommand{\node}{y}
\newcommand{\nodenumber}{n}
\newcommand{\nodeindex}{j}
\newcommand{\nodeindexalt}{i}
\newcommand{\nodeindexaltb}{k}
\newcommand{\firstnodeindex}{1}
\newcommand{\secondnodeindex}{2}
\newcommand{\thirdnodeindex}{3}
\newcommand{\lastnodeindex}{\nodenumber}
\newcommand{\collider}{c}
\newcommand{\collidernumber}{n_c}
\newcommand{\ancestornumber}{n_a}

\newcommand{\path}{\pi}
\newcommand{\pathfirstnode}{0}
\newcommand{\pathsecondnode}{1}
\newcommand{\pathsecondlastnode}{l-1}
\newcommand{\pathlastnode}{l}
\newcommand{\pathlength}{\ell}
\newcommand{\pathnodeindex}{p}
\newcommand{\pathnodeindexalt}{q}
\newcommand{\pathapex}{(a)}
\newcommand{\pathapexalt}{(b)}

\newcommand{\ov}{\overline}

\newcommand{\parent}[2]{pa_{#1}\left(#2\right)}
\newcommand{\child}[2]{ch_{#1}\left(#2\right)}
\newcommand{\ancestor}[2]{an_{#1}\left(#2\right)}
\newcommand{\descendant}[2]{de_{#1}\left(#2\right)}

\newcommand{\cov}{R}
\newcommand{\timelag}{\tau}
\newcommand{\E}{\mathbf{E}}
\newcommand{\Ztrans}{\mathcal{Z}}

\newcommand{\graph}{G}
\newcommand{\vertexset}{V}
\newcommand{\edgeset}{E}

\newcommand{\subsetvertex}{J}
\newcommand{\subsetvertexalt}{I}
\newcommand{\subsetvertexsep}{Z}

\newcommand{\imap}{\mathcal{I}}
\newcommand{\granger}{\mathcal{G}}
\newcommand{\csep}{\mathcal{C}}

\newcommand{\colliderset}{C}

\newcommand{\PSD}{\Phi}

\newcommand{\TheoremMatSalMarkovBlanketnoncausal}{Theorem~27 in \cite{MatSal12}}
\newcommand{\LemmaMatSalWienerUncorrelated}{Lemma~26 in \cite{MatSal12}}

\newcommand{\childrenset}{\mathcal{C}}
\newcommand{\parentset}{\mathcal{P}}
\newcommand{\coparentset}{\mathcal{K}}
\newcommand{\tf}{H}

\newcommand{\lemmarestrictiontoancestors}{Lemma~6 in \cite{MatSal14}}
\newcommand{\lemmaenlargeseparatedsets}{Lemma~7 in \cite{MatSal14}} 

\newcommand{\wiener}{W}
\newcommand{\Wiener}[3][]{
	\wiener_{
		#2
		\ifx&#1&
		\else
		,[#1]
		\fi
		|#3}
	}

\newcommand{\LDIMGraph}{Direct Feedthrough Graphical Representation }
\newcommand{\LDIMgraph}{direct feedthrough graphical representation }

\begin{frontmatter}

\title{An algorithm for reconstruction of triangle-free linear dynamic networks with verification of correctness} 

\author{Mihaela Dimovska \& Donatello Materassi} 
\address{Department of Electrical and Computer Engineering,\\
		University of Minnesota, 200 Union St SE, 55455, Minneapolis (MN)}

\begin{abstract}                
Reconstructing a network of dynamic systems from observational data is an active area of research. 
Many approaches guarantee a consistent reconstruction under the relatively strong assumption that the network dynamics is governed by strictly causal transfer functions.
However, in many practical scenarios, strictly causal models are not adequate to describe the system and it is necessary to consider models with dynamics that include direct feedthrough terms.
In presence of direct feedthroughs, guaranteeing a consistent reconstruction is a more challenging task.
Indeed, under no additional assumptions on the network, we prove that, even in the limit of infinite data, any reconstruction method is susceptible to inferring edges that do not exist in the true network (false positives) or not detecting edges that are present in the network (false negative). 
However, for a class of triangle-free networks introduced in this article, some consistency guarantees can be provided. We present a method that either exactly recovers the topology of a triangle-free network certifying its correctness or outputs a graph that is sparser than the topology of the actual network, specifying that such a graph has no false positives, but there are false negatives.
\end{abstract}


\end{frontmatter}

\section{Introduction}

Reconstructing a network of dynamic systems from observational data, while providing some provable guarantees about the resulting reconstructed graph is an active area of research.
Many of the results providing consistent reconstruction generally assume that the system dynamics is strictly causal \cite{granger1969, yue2017linear, GonWar08, etesami2014directed}. However, in practice, the necessity of using models with non-necessarily strictly causal dynamics arises in many areas, such as biology \cite{schiatti2015extended, faes2015linear}, neuroscience \cite{seth2015granger} or finance \cite{materassi2009unveiling}.
For this reason, there are methods that, at least in the linear case, try to deal with direct feedthroughs, too. Some of these methods are often based on extension of the notion of Granger causality \cite{granger1969} by adjusting and controlling for intermediate causes \cite{quinn2011estimating, schiatti2015extended}. 
The main drawback of these extensions of Granger causality is that, even in the limit of infinite data, they might output a network topology with both false positives (an inferred edge that is not in the actual network) and false negatives (a missing edge that is in the actual network). Furthermore, even if the output is correct, these methods do not certify its correctness. 
On the other hand, there are methods, such as \cite{runge2015identifying, dimovska2017granger} that borrow tools from the theory of causal inference in graphical models \cite{pearl2009causality, spirtes2000causation}.
Even though these methods still cannot certify the correctness of their output, some of them have the attractive feature of guaranteeing no false positives \cite{dimovska2017granger}.
Guaranteeing that there are also no false negatives in the recovered network topology is a more challenging task.
Indeed, in the presence of direct feedthroughs, with no additional assumptions on the network structure, we show that no method can provide guarantees for no false negatives and no false positives and we illustrate this via an example in Section~\ref{sec:no false negatives}.

Then, a fundamental question to be investigated is the determination of subclasses of networks for which there are topology identification methods guaranteeing exact reconstruction with the additional goal of making such subclasses as extensive as possible. 
In the domain of graphical models, the question of under what assumptions an exact reconstruction is possible has been tackled from different perpectives.
Typically, assumptions are made directly on the network dynamics. For example, graphical models are generally defined using static operators instead of dynamic ones, such as transfer functions, and the extension to the dynamic case is often non-trivial. Another common assumption in the area of graphical models is referred to as ``faithfulness'' and boils down to the absence of cancellations in the paths of the network graph \cite{spirtes2000causation, uhler2013geometry}.
Attempts to weaken the condition of faithfulness have been described, for example, in \cite{park2016identifiability}, but still result in some forms of assumptions on the network dynamics.
In this work, instead, the assumptions are made directly on the graph structure.
Furthermore, the method provided in \cite{park2016identifiability} is of combinatorial complexity in the number of nodes of the network, for any kind of network. 
In this work, our algorithm has the appealing property of running in polynomial time, at least for sparse networks.

The main results in this article can be seen as a step towards finding classes of networks for which it is possible to give guarantees of an exact reconstruction.
We define the subclass of triangle-free unidirectional networks and provide a method that either reconstructs the correct topology while certifying its correctness, or outputs a topology with no false positives, specifying though that there definitely are false negatives.
The article  is organized as follows: 
In Section~\ref{sec:background} we provide the necessary background terminology and information to derive the main results; in Section~\ref{sec:no false negatives} we provide a counterexample that shows that no algorithm can guarantee an exact reconstruction without additional network assumptions; in Section~\ref{sec:main results} we present the main results; in Section~\ref{sec:examples} we provide examples that demonstrate the new algorithm under different input scenarios; in Section~\ref{sec:conclusion} we present the conclusions of this work.

\section{Preliminary notions}\label{sec:background}
The goal of this section is first to provide some background about graph theory concepts which are widely used in the literature about graphical models, such as chains and colliders in paths, Markov blanket of a node in a directed graph and moral graphs \cite{koller2009probabilistic}.
Then, we will introduce the class of network models called Linear Dynamic Influence Models (LDMIs) which will be the focus of our work \cite{materassi2019signal}.

\subsection{Graph theory concepts}

First, we recall the concepts of directed/undirected graphs. 

\begin{dfn}[Directed and Undirected Graphs]~\\
	A directed (undirected) graph $G$ is a pair $(V,E)$, where $V=\{y_{1}, y_{2}, ..., y_{n}\}$  is a set of vertices or nodes and $E\subseteq V^{2}$ is a set of edges or arcs, which are ordered (unordered) subsets of pairs of elements in~$V$. 
\end{dfn}
It is possible to associate an undirected graph to any directed graph by removing the orientation of its links.
\begin{dfn}[Skeleton ~\cite{Pea88}]
	Given a directed graph $\graph=(\vertexset,\edgeset)$, we define its ``skeleton'' (or ``topology'') as the undirected graph $(\vertexset,\ov \edgeset)$ obtained by removing the orientation of its edges. 
\end{dfn}
On an undirected graph we recall that $y_{j}$ and $y_{i}$ are neighbors if and only if $\{y_{i}, y_{j}\} \in \overline{E}$. We denote the set of neighbors of a node $y_{i}$ as $N(y_{i})$.  
For undirected graphs defined over the same vertex set, we introduce a partial order relation which is based on their adjacencies. 
\begin{dfn}[Upper/Lower Bound of an undirected graph]
Let $G_{1} = (V, E_{1})$ and $G_{2} = (V, E_{2})$ be two undirected graphs with the same vertex set. If $E_{1} \subseteq E_{2}$, we say that $G_{2}$ is an upper-bound for $G_{1}$, or equivalently that $G_{1}$ is a lower-bound for $G_{2}$.
\end{dfn}

On a directed graph we also recall what ``chains'' (or ``directed paths'') are (see~\cite{Die12} or \cite{Pea88} for more details and for the formal definition). 
	A {\it chain} starting from node
	$y_{i}$
	and ending in node
	$y_{j}$
	is an ordered set of edges in $\edgeset$
	$(\,(y_{\path_{1}},y_{\path_{2}}),\ldots, (y_{\path_{\pathlength-1}}, y_{\path_{\pathlength}})\,)$
	where
	$y_{\path_{1}}=y_{i}$,
	$y_{\path_{\pathlength}}=y_{j}$.
We also use the standard notions of parents, children, ancestors, and descendants in a directed graph $G$.
	A vertex $y_i$ is a {\it parent} of a vertex $y_j$ if there is a directed edge from $y_i$ to $y_j$.
	In such a case, we also say that $y_j$ is a {\it child} of $y_i$.
	Furthermore, $y_i$ is an  {\it ancestor} of  $y_j$ if $y_i=y_j$ or there is a chain from $y_i$ to $y_j$.
	In such a case, we also say that $y_j$ is a {\it descendant} of $y_i$. 
	A {\it path} between two vertices $y_{i}$ and $y_{j}$	is an ordered sequence of ordered pairs of nodes 
	$\left((y_{\pi_{1}}, y_{\pi_{2}}), (y_{\pi_{3}}, y_{\pi_{4}}), ..., (y_{\pi_{m-1}}, y_{\pi_{m}})\right) $, 
	with $y_{\pi_{1}} =y_{i}, y_{\pi_{m}}=y_{j}$, 
	$y_{\pi_{1}} \neq y_{\pi_{2}} ...\neq y_{\pi_{m}} $,
	such that $(y_{\pi_{l}}, y_{\pi_{l+1}}) \in E$ or $(y_{\pi_{l+1}}, y_{\pi_{l}}) \in E$, for all $l=1...,m-1$.

 On a given path we define the notion of colliders.
 \begin{dfn}[Colliders and coparents]
    A path has a {\it collider} at
 	$y_{k}$
 	if
 	$y_{i}$
 	and
 	$y_{j}$
 	are both parents of
 	$y_{k}$
 	(that is
 	$y_{i} \rightarrow y_{k} \leftarrow y_{j}$
 	appears in the path). We call $y_{i}$ and $y_{j}$ {\it coparents}. 
 \end{dfn}

 \begin{dfn}[Moral graph \cite{koller2009probabilistic}] Given an oriented graph $\graph=(\vertexset,\edgeset)$, its moral graph is the undirected graph $G_{M} = (V, E_{M})$ where  $ \left\{ {y_i, y_j }\right\} \in E_{M}$ if $(y_i, y_j) \in E$, or $(y_j, y_i) \in E$, or $y_i$ and $y_j$ are coparents in $G$. 
 \end{dfn}
 
 \begin{dfn}[Markov Blanket \cite{koller2009probabilistic}]
 Let $G = (V,E)$ be an oriented graph and let $G_{M} = (V, E_{M})$ be its moral graph. 
 The Markov Blanket of a node $y_j \in V$, denoted by $MB(y_{j})$, is the set of neighbors of $y_j$ in the moral graph. 
 \end{dfn}
 \begin{figure}[h!]
 	\centering
 	\begin{tabular}{c}
 	\includegraphics[width=0.9\columnwidth]{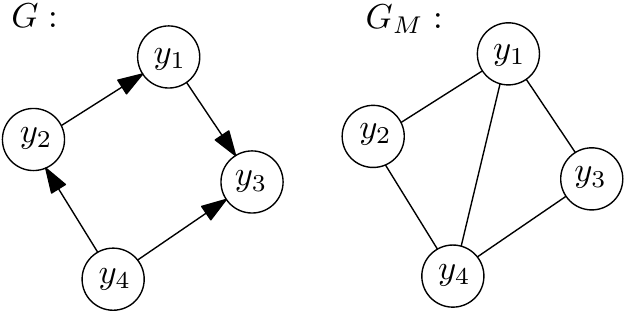}
 	\end{tabular}
 	\caption{A directed graph $G$ and its moral graph $G_{M}$. Notice that $G_{M}$ has an edge between the co-parents $y_{1}$ and $y_{4}$. The Markov Blanket of $y_1$ is $MB(y_1) = \left\lbrace y_2, y_3, y_4 \right\rbrace$. \label{fig:markov blanket}}
 \end{figure}
 An example of a moral graph together with the Markov blanket of a node from that graph is given in Figure~\ref{fig:markov blanket}. 
 
 Next, we introduce the class of systems that we consider in this work. 
 
\subsection{Generative class of networks: Linear Dynamic Influence Models}

Linear Dynamic Influence Models (LDIMs) are a class of networks representing input/output relations among stochastic processes \cite{materassi2012problem, dimovska2017granger}.

\begin{dfn}[Linear Dynamic Influence Model]

A Linear Dynamic Influence Model $\mathcal{G}$ is a pair $(H(z),e)$ where
	\begin{itemize}
		\item	$e=(e_1, ..., e_n)^{T}$ is a vector of $N$ scalar processes $e_{1}$, ... $e_{n}$, such that 
		$\Phi_{e}(z)$, the Power Spectral Density (PSD) of $e$,
		is real-rational and diagonal, namely $\Phi_{e_{i}e_{j}}=0$ for $i\neq j$.
		\item	$H(z)$ is an $n\times n$ real-rational transfer matrix.    $H(z)$ is termed as the ``dynamics'' of the LDIM.
	\end{itemize}
	The output processes $\{y_j\}_{j=1}^{n}$ of the LDIM are defined as
	$y_j=e_j+\sum_{i=1}^{n}H_{ji}(z)y_i$,
	or in a more compact way 
	$
		y=e+H(z)y
	$,
	where $y=(y_{1}, ..., y_{n})^{T}$.

\end{dfn}
Every LDIM admits a graphical representation as a directed graph, called the causal graph of the LDIM. 

\begin{dfn}[Causal Graph of a LDIM]
Let $\mathcal{G}$ be an LDIM. Let $G=(V,E)$ be a directed graph defined as follows.
$V=\{y_{1},...,y_{n}\}$ is the set of the output processes of the LDIM
and the set of edges $E$ contains 
$(y_{i}, y_{j})$ if and only if $H_{ji} \neq 0$.
We refer to $G$ as the causal graph of the LDIM.
\end{dfn}

Throughout this work we will refer to nodes and edges of the LDIM, meaning nodes and edges of the causal graph of the LDIM.

%

\begin{dfn}[Unidirectional Triangle-free LDIMs]
An LDIM is unidirectional triangle-free if the causal graph of the LDIM does not have any  loops of length $2$ and its skeleton does not contain any triangles.  
\end{dfn}

We note that the class of triangle-free networks contains some important classes of networks, for example, trees and polytrees \cite{sepehr2016inferring, sepehr2019blind}.
Next, we define the concept of Wiener separation \cite{materassi2019signal} that will be used to decide if two nodes are directly connected or not. 

\begin{dfn}[Wiener separation]\cite{materassi2019signal}
 Let $(H(z),e)$ be a LDIM. We say that the process $y_{j}$ is Wiener separated from $y_{i}$ given a set of processes $S \subseteq y$ if the Non-Causal (Causal) Wiener filter \cite{Wie49} estimating $y_{j}$ from $y_{i} \cup S$ has zero-entry associated with $y_{i}$. 
\end{dfn}

Note that the concept of Wiener separation can be defined using either the causal Wiener filter or its non-causal counterpart. Depending on whether we use the causal or non-causal Wiener filter, we denote Wiener separation as 
$cwsep(y_{j}, S, y_{i})$ and $wsep(y_{j}, S, y_{i})$, respectively. 
Similarly, if Wiener separation does not hold, we use the notation
$\neg cwsep(y_{j}, S, y_{i})$ and $\neg wsep(y_{j}, S, y_{i})$.


\section*{Network Skeleton Reconstruction Problem}

From the PSD matrix of the output processes $y$ of a LDIM, determine the skeleton of the network.  \\

\section{Counterexample for accuracy certificate without network assumptions}\label{sec:no false negatives}
In this section we show that, without any additional assumptions, 
the Network Skeleton Reconstruction Problem is not well-posed, namely it does not admit, in general, a unique solution.

Indeed, we provide a counterexample showing that without any additional assumptions on the network of the system, no method can guarantee an exact reconstruction of the skeleton of the network from observational data. 

Consider an LDIM $\mathcal{G}_{1} = (H_{1}(z), E_{1})$ with transfer function 
\[H_{1}(z) = 
\begin{bmatrix}
0 & 0 & 0\\
a & 0 & 0 \\
c & b & 0
\end{bmatrix}  \,,
\quad \Phi_{E_{1}E_{1}}=\begin{bmatrix}
1 & 0 & 0\\
0 & 1 & 0 \\
0 & 0 & 1
\end{bmatrix}
\]
where $a,b,c \in \mathcal{R}$, $a,b,c \neq 0$. The causal graph of this LDIM is shown in Figure~\ref{fig:triangle}(a).

Alternatively, consider the LDIM $\mathcal{G}_{2} = (H_{2}(z), E_{2})$:
\[H_{2}(z) = 
\begin{bmatrix}
0 & 0 & 0\\
a & 0 & \frac{b}{b^2+1} \\
0 & 0 & 0 
\end{bmatrix}; \quad \Phi_{E_{2}E_{2}} = 
 \begin{bmatrix}
1 & 0 & 0\\
0 & \frac{1}{b^2+1} & 0 \\
0 & 0 & b^2+1 
\end{bmatrix} \,.\]
The causal graph of this LDIM is shown in Figure~\ref{fig:triangle}(b). 
Note for $c=-a\cdot b$ the PSDs of $\mathcal{G}_{1}$ and $\mathcal{G}_{2}$ are equal: 
 \[\Phi_{Y_{1}Y_{1}} = \Phi_{Y_{2}Y_{2}} =  
 \begin{bmatrix}
1 & a & 0\\
a & a^2+1 & b \\
0 & b & b^2+1 
\end{bmatrix} \,.\]

As the PSDs of both systems are the same, no method can distinguish these two systems from observational data only, thus no method can guarantee an exact reconstruction from observational data for networks with general dynamics and topology.  


 \begin{figure}[h!]
	\centering
	\begin{tabular}{cc}
	\includegraphics[width=0.3\columnwidth]{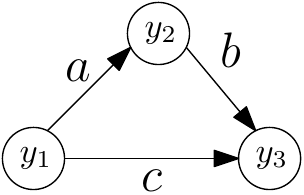} 
	&
	\includegraphics[width=0.3\columnwidth]{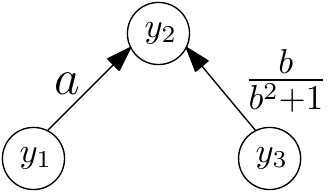}	
	\\ 
	(a) & (b) \\ 
	
	\end{tabular}\caption{(a) The causal graph of $\mathcal{G}_{1}$ (b) The causal graph of $\mathcal{G}_{2}$. We cannot distinguish between these two causal graphs given the dynamics of  $\mathcal{G}_{1}$ and  $\mathcal{G}_{2}$. \label{fig:triangle}}
\end{figure}  

The reason why these two systems cannot be distinguished is because of path cancellations occuring in $\mathcal{G}_{1}$. Namely, the influence of $y_{1}$ to $y_{3}$ through the direct edge $y_{1}\to y_{3}$ cancels the influence of $y_{1}$ to $y_{3}$ through the path $y_{1} \to y_{2} \to y_{3}$. Thus, in $\Phi_{Y_{1}Y_{1}}$ we get zeros in the entries $\{1,3\}$ and $\{3,1\}$. This enables us to construct another system, with a sparser causal graph, that has the same PSD as $\mathcal{G}_{1}$.

The example discussed above leads to the question of whether, under some additional assumptions, we can provide some guarantees for both no false positives and no false negatives in the identification of graph edges. 
Furthermore we would like those additonal assumptions to be imposed on the underlying structure of the network, as we would like to keep the assumptions on the dynamics mild (namely only algebraic loops shouldn't be allowed). Further, the role of these graphical assumptions should be to prevent the occurence of situations like the one in the previous example, where we have the cancellation of the effect of a parent node on its child because of the intermediate action of a coparent. In the following, by assuming that the LDIM is unidirectial triangle-free we will achieve this goal. 

\section{Main Results}\label{sec:main results}

In this section we present an algorithm with provable guarantees for the correct reconstruction of the skeleton of a unidirectional triangle-free LDIM. We start by providing a result that enables to infer an upper bound of the skeleton of a triangle-free LDIM. The result is summarized in the following Lemma. 

\begin{lem}[Moral graph identification]\label{lem:moral graph}
 Consider the unidirectional triangle-free LDIM $(H(z),e)$ having $G=(V,E)$ as its causal graph. Let $G_{M}$ be the moral graph of $G$.
 For all unordered pairs $y_{i}$ and $y_{j}$ let $S_{ij}=y\setminus \{y_{i}, y_{j}\}$. Define the graph $\overline{G}=(V,\overline{E})$ where
 $\{y_{i},y_{i}\}\in \overline{E}$ if and only if $\neg wsep(y_{j}, S, y_{i})$.
 We have that $\overline{G}$ is a lower bound for $G_{M}$ and an upper bound for the skeleton of $G$. 
\end{lem}

\begin{proof}
 From Theorem~30 in \cite{materassi2012problem} we have that the non-causal Wiener estimator of $y_{j}$ from $y_{i}\cup S$ is of the form:
 \[
  \hat {y_{j}} = \sum_{y_{k} \in S}{W_{jk}(z)y_{k}} + W_{ji}y_{j} \,.
 \]
 Note that $W_{ji}(z)$ denotes the component of the Wiener filter corresponding to $y_{i}$, when we estimate $y_{j}$ from $y_{i} \cup S$. 
 From Lemma~31 in \cite{materassi2012problem} we have that the component $W_{ji}(z)$ is a sum of three terms, namely $W_{ji}(z) = C_{ji}(z) + P_{ji}(z) + K_{ji}(z)$, where $C_{ji}(z) \neq 0$ if and only if $y_{i}$ is a child of $y_{j}$, $P_{ji}(z) \neq 0 $ if and only if $y_{i}$ is a parent of $y_{j}$, and $K_{ji}(z) \neq 0$ implies that $y_{i}$ and $y_{j}$ are co-parents of a node.
 Thus, if two nodes $y_{i}$ and $y_{j}$ are connected in $\overline{G}$ then they are connected in $G_{M}$, proving that $\overline{G}$ is a lower bound for $G_{M}$.
 
 Now we prove that the skeleton of $G$ is a lower bound for $\overline{G}$.
 If $\{y_{i},y_{j}\}$ is in the skeleton of $G$ then either $C_{ji}(z)$ or $P_{ji}(z)$ are different from $0$.
 Since the input graph is unidirectional triangle-free, at most one of the terms $C_{ji}(z)$, $P_{ji}(z)$ or $K_{ji}$ is different from zero, hence $W_{ji}(z)$ is different from zero.
\end{proof}
 
 We further illustrate the above Lemma with a simple example. Consider a dynamic system with a network as shown in Figure~\ref{fig:coparents_cancellation}(a). Let the dynamics be given by the direct feedthrough transfer functions 
 $h_{31} = -a, h_{32} = b, h_{41} = a, h_{42}=b$, with $a,b \neq 0$. Then, the non-causal Wiener filter component corresponding to $y_{2}$, when estimating $y_{1}$ given $S=\{y_{3}, y_{4}\}$ and $y_{2}$, is $0$. Thus, while the true moral graph of the network is the graph shown in Figure~\ref{fig:coparents_cancellation}(b), the graph that is inferred, by following the result of Lemma~\ref{lem:moral graph}, is the one shown in Figure~\ref{fig:coparents_cancellation}(c).

 \begin{figure}[h!]
	\centering
	\begin{tabular}{ccc}
	\includegraphics[width=0.29\columnwidth]{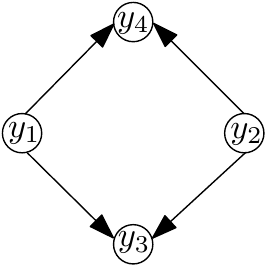} 
	&
	\includegraphics[width=0.29\columnwidth]{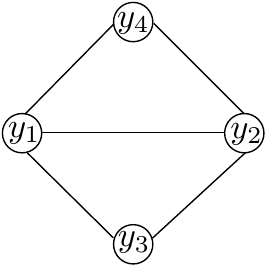}
	&
	\includegraphics[width=0.29\columnwidth]{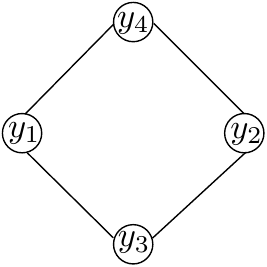}
	\\ 
	(a) & (b) & (c) 
	
	\end{tabular}\caption{(a) Network of a dynamic system where, due to the specific transfer functions, the inferred moral graph by using Lemma~\ref{lem:moral graph} is the one shown in (c), while its true moral graph is the one shown in (b).  \label{fig:coparents_cancellation}}
\end{figure}  
 
 Thus, with Lemma~\ref{lem:moral graph} we obtain an undirected graph that is an upper bound to the true skeleton of the network- namely it contains some false positive edges linking the coparents.
 In order to remove these false positive edges, we exploit a variation of Mixed Delay (MD) algorithm from \cite{dimovska2017granger}. 
 MD algorithm is a network reconstruction algorithm which starts from the complete graph and provably outputs a lower bound for the skeleton of a LDIM by sequentially testing all the edges. 
Here, we recall the edge-removal properties of the MD algorithm that guarantee that it will not infer any false positive edges for the reconstructed skeleton. 
Namely, MD algorithm relies on the following facts: if there is no edge 
 $\{y_{i},y_{j}\}$ in the skeleton of the causal graph of the LDIM then 
\begin{enumerate}
 \item[\tt MD1] There exist two sets 
 $S_{ji}^{+} \subseteq y\setminus \{y_{i}, y_{j}\}$
 and $S_{ji}^{-} \subseteq \frac{1}{z}y$ such that the $y_{i}$ component of the causal Wiener filter estimating $y_{j}$ given $y_{i} \cup S_{ji}^{+} \cup S_{ji}^{-}$ is strictly causal and
  \item[\tt MD2] There exist 
  $S_{c}^{i} \subseteq y\setminus \{y_{j}\} $ and $S_{s}^{i} \subseteq \frac{1}{z}y \setminus \frac{1}{z}y_{i}$
  such that 
  $cwsep(y_{j}, S_{c}^{i}\cup S_{s}^{i}, \frac{1}{z}y_{i})$ and 
  \item[\tt MD3] There exist 
  $S_{c}^{j} \subseteq y\setminus \{y_{i}\} $ and $S_{s}^{j} \subseteq \frac{1}{z}y \setminus \frac{1}{z}y_{j}$ such that 
  $cwsep(y_{i}, S_{c}^{j}\cup S_{s}^{j}, \frac{1}{z}y_{j})$ .
  \end{enumerate}
Because of implications 1), 2) and 3), if there is no edge  $\{y_{i}, y_{j}\}$ in the skeleton of the LDIM, the sets described in each of the steps above are guaranteed to be found, so the MD algorithm correctly infers that there is no edge between those two nodes.

The main drawback of MD algorithm is that its computational complexity is combinatorial with the number of LDIM nodes.
 Indeed, in order to determine if the edge 
 $\{y_{i},y_{j}\}$ belongs to its output, MD algorithm needs to run a search among all subsets of $y$.
 However, following the proofs of Theorem~4.1 and Theorem~4.2 in \cite{dimovska2017granger}, we have the following observation. If there is no edge between $y_{i}$ and $y_{j}$, then separating sets with properties as described in each of the three steps above, can be constructed using only nodes which are parents of either $y_{i}$ or $y_{j}$.  
Since MD has no a-priori knowledge of the causal graph, it cannot limit the search to 
the parents of either $y_{i}$ or $y_{j}$ and instead extends it to all nodes.
However, note that Lemma~\ref{lem:moral graph} enables us to find an upper bound $\overline{G}$ of the true skeleton $G$. For every node $y_{j} \in \overline{G}$ the set of parents of $y_{j}$ are contained in the set of neighbors of $y_{j}$, $N(y_{j})$. Thus, we can use the following improved variation of the MD-algorithm.
If there is no edge between $y_{i}$ and $y_{j}$ in the input LDIM, then: 
 \begin{enumerate}
  \item[\tt MD1+] There exists $S_{ji}^{+} \subseteq N(y_{i}) \cup N(y_{j})$ and $S^{-} \subseteq \frac{1}{z} N(y_{i}) \cup \frac{1}{z} N(y_{j}) $ such that the $y_{i}$ component of the causal Wiener filter estimating $y_{j}$ given $y_{i} \cup S_{ji}^{+} \cup S_{ji}^{-}$ is strictly causal and
  \item[\tt MD2+] There exists $S_{c}^{i} \subseteq N(y_{i}) \cup N(y_{j}) $ 
  and 
  $S_{s}^{i} \subseteq \{\frac{1}{z}N(y_{i}) \cup \frac{1}{z}N(y_{j})\}$ such that 
  $cwsep(y_{j}, S_{c}^{i}\cup S_{s}^{i}, \frac{1}{z}y_{i})$ and 
  \item[\tt MD1+] There exists $S_{c}^{j} \subseteq N(y_{i}) \cup N(y_{j})$ and 
  $S_{s}^{j} \subseteq \{ \frac{1}{z}N(y_{i}) \cup \frac{1}{z}N(y_{j}) \} $ such that 
  $cwsep(y_{i}, S_{c}^{j} \cup S_{s}^{j}, \frac{1}{z}y_{j})$ 
 \end{enumerate}
Now we are ready to state the main result of this article in the next theorem.

 \begin{thm}
 Let $\mathcal{G}$ be a LDIM. 
Let $\overline{G}$ be an undirected graph that is an upper bound of the skeleton of the LDIM and a lower bound of the moral graph of that skeleton. Let $(y_{i}, y_{j}, y_{k})$ be an unordered triplet denoting three nodes of a triangle in the graph $\overline{G}$. If by using steps {\tt MD1+, MD2+ and MD3+} exactly one edge from the triangle $(y_{i}, y_{j}, y_{j})$ is removed, then that edge was a co-parent edge. If more than one edge is removed, then one of those removals is a false negative. 
\end{thm}

\begin{proof}


As the graph $\overline{G}$ is upper bound for the skeleton and lower bound for the moral graph, the only edges that are in $\overline{G}$ but are not in the skeleton are edges linking coparents. 
Further, as the skeleton is triangle free, two out of the triangle triplet nodes $\left(y_{i}, y_{j}, y_{k}\right)$ must be coparents in the LDIM.  

Using the steps {\tt MD1+, MD2+ and MD3+}, we try to remove an edge from the triangle consisting of the nodes $\left(y_{i}, y_{j}, y_{k}\right)$. 
As the steps in $\tt MD+$ guarantee no false positives, the edge linking the coparents will be removed. Without loss of generality, let the coparents be $y_{i}$ and $y_{j}$. If the only edge we are able to remove is the $y_{i}-y_{j}$ edge, then we have removed only the coparent edge. If we can remove at least one more edge, then we a removing a link that is indeed present in the true skeleton. 
\end{proof}

Now we provide the detailed steps of the algorithm that  either outputs the correct skeleton, with certified correctness, or outputs a flagged sparser network, stating that the sparser network is not the true one. \\
\newpage
\begin{framed}
\textbf{Unidirectional Triangle-Free Skeleton Reconstruction (UTF-SR)} \begin{itemize} 
 \item[\tt 1.]  For every $y_{i}, y_{j}$ test $wsep(y_{i}, S, y_{j})$, with $S=y\setminus\{y_{i}, y_{j}\}$, using the non-causal Wiener filter. This step outputs $\overline{G}$ which is an upper-bound for the true skeleton and a lower bound for the moral graph.
 \item[\tt 2.] For every edge $y_{i}-y_{j} \in \overline{G}$ that is part of a triangle, test if we can remove the edge by using the 3 steps of the $\tt MD+$ algorithm.
    \begin{itemize}
     \item[\tt 3.1]  If exactly one edge can be removed from each triangle then we output the correct skeleton.
    \item[\tt 3.2] Else if more than one edge can be removed from a triangle, we output a flagged skeleton indicating that we have found a skeleton that is a lower bound for the true one. 
    \end{itemize}

\end{itemize}\end{framed}

\section{Examples}~\label{sec:examples}

In this section we demonstrate the UTF Skeleton Reconstruction algorithm in several scenarios. In the first case we illustrate an example in which both the moral graph and the skeleton are retrieved correctly. The second case illustrates a situation where the algorithm outputs a flagged sparser skeleton; and the third case illustrates how the first step of the algorithms helps in identifying the correct skeleton. Lastly, we also demonstrate the UTF SR on an LDIM with feedback loop of length four. 

\begin{example}\label{ex:1}
 Consider the LDIM with causal graph as in Figure~\ref{fig:unfaithfulness_diamond}(a), with transfer functions $h_{14}=1, h_{21}=1, h_{32}=1, h_{34}=1$ and $\Phi_{EE}=I$. 
 Then, the first step in UTF-SR retrieves the correct moral graph shown in Figure~\ref{fig:unfaithfulness_diamond}(b). As all the edges are involved in a triangle, every edge is tested for removal with the $\tt MD+$ edge removal steps. With those steps, we are able to remove exactly one edge from each triangle, namely the edge $y_{2}-y_{4}$ as $cwsep(y_{2}, S=y_{1}, y_{4})$.
 Thus, we infer the correct skeleton shown in Figure\ref{fig:unfaithfulness_diamond}(c). 
\end{example}

\begin{example}\label{ex:2}

Next, consider an LDIM with the same causal graph as the LDIM from the previous example. However, in this example, the transfer functions of the LDIM are $h_{14}=2, h_{21}=2, h_{32}=2, h_{34}=-8$. 
Though with the first step of the UTF-SR algorithm we again get the correct moral graph, running $\tt MD+$ on every triangle edge, we get that $cwsep(y_{3}, \emptyset, y_{4})$, despite $h_{34} \neq 0$. We can easily check that the $\tt MD+$ edge-removal steps will also remove the edge between $y_{2}$ and $y_{4}$. Thus, we are able to remove two edges from a triangle and UTF-SR outputs the skeleton in Figure~\ref{fig:unfaithfulness_diamond}(d) which is sparser than the skeleton of the true causal graph and it is labeled as such. 
\end{example}

 \begin{figure}[h!]
	\centering
	\begin{tabular}{cccc}
	\includegraphics[width=0.2\columnwidth]{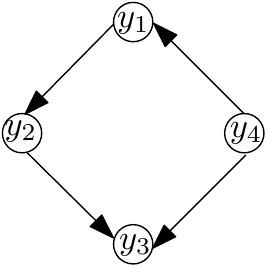} & 
	\includegraphics[width=0.2\columnwidth]{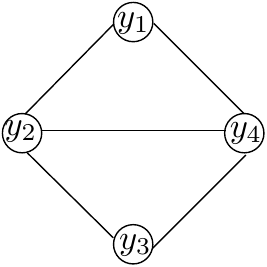} 
	& 
	\includegraphics[width=0.2\columnwidth]{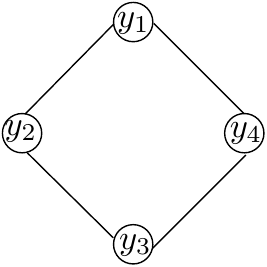} &
	\includegraphics[width=0.2\columnwidth]{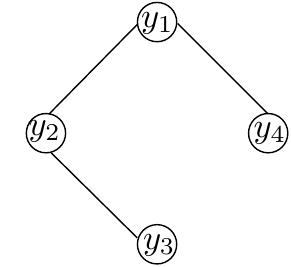} \\
	(a) & (b) & (c) & (d)
	\end{tabular}\caption{(a) The causal graph of the LDIM from Example~\ref{ex:1} and Example~\ref{ex:2}. (b) The true moral graph of the LDIMs from  Example~\ref{ex:1} and Example~\ref{ex:2}. (c) The true skeleton of the causal graph of the LDIMS from  Example~\ref{ex:1} and Example~\ref{ex:2}. (d) The sparser skeleton that is the output when applying the algorithm to the LDIM in Example~\ref{ex:2}.  \label{fig:unfaithfulness_diamond}}
\end{figure}

\begin{example}\label{ex:3}
Consider a LDIM with causal graph as in Figure~\ref{fig:simple_kin_md_third_case}(a), with transfer function:
 \[
  H(z) = \begin{pmatrix}0 & 0 & 0 & 1 & 0\\
2 & 0 & 0 & 0 & 0\\
0 & 3 & 0 & -6 & 0\\
0 & 0 & 0 & 0 & 0\\
0 & c & 0 & 6 & 0\end{pmatrix} \,, \quad \Phi_{EE}=I \,.
\]

It is easy to check that the moral graph that we obtain with the first step in the triangle-free MD algorithm is the one shown in 
 Figure~\ref{fig:simple_kin_md_third_case}(b), as the $y_{4}$ component of the non-causal Wiener filter estimating $y_{2}$ given $S=\{y_{4}, y_{1}, y_{3}, y_{5} \}$ is zero. In the graph retrieved with the first step there are no triangles to check so we get the correct skeleton. \\
 However, note that if we were to actually pass the moral graph as input to the next steps, then we wouldn't have been able to retrieve the correct skeleton. Indeed, note that though $h_{34} \neq 0$ we have that $cwsep(y_{3},\emptyset, y_{4})$, so the edge-removal steps of $\tt MD+$ will remove the edge $y_{3}-y_{4}$, leading to a false negative. Thus, the first step is not just making the algorithm more efficient, but it can also help certify the correct skeleton. 
 
  \begin{figure}[h!]
	\centering
	\begin{tabular}{cc}
	\includegraphics[width=0.35\columnwidth]{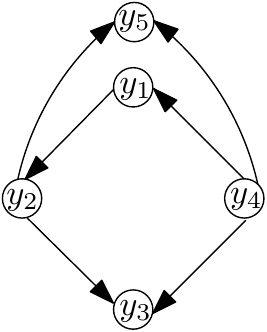} & 
	\includegraphics[width=0.35\columnwidth]{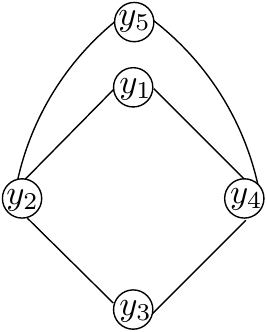} \\
	(a) & (b) 
	\end{tabular}\caption{(a) The causal graph of the LDIM from Example~\ref{ex:3} (b) The correct skeleton of the LDIM from Example~\ref{ex:3}.  \label{fig:simple_kin_md_third_case}}
\end{figure}  
 
\end{example}

\begin{example}\label{ex:4}
Consider an LDIM with causal graph shown in Figure~\ref{fig:cycle}(a) and transfer function:
\[H(z) = 
\begin{bmatrix}
0 & 0 &0 & 3 & 4  \\
1 & 0 & 0 & 0 & 0 \\
0 & \frac{1}{z} & 0 & 0 & 0\\
0 & 0 &  2 & 0 & 0 \\
0 & 0 &  0 & 3 & 0 \\
0 & 0 &  0 & 0 & 0 \\
\end{bmatrix}  \,,
\quad \Phi_{EE}=\mathcal{I} \,.
\]
 
With the first step of the triangle-free algorithm we obtain the moral graph of the LDIM shown in Figure~\ref{fig:cycle}(b). Note that the only triangle is the one in the nodes $y_{1}, y_{5}$ and $y_{4}$. Despite the cyclic causal graph of this LDIM, using the MD algorithm we can infer that there is no edge between $y_{5}$ and $y_{4}$. Indeed, the following causal Wiener-separating statements hold:
\begin{itemize}
 \item $cwsep(y_{5}, \emptyset, y_{4})$
 \item $cwsep(y_{5}, \emptyset, \frac{1}{z}y_{4})$
 \item $cwsep(y_{4}, \emptyset, \frac{1}{z}y_{5})$.
\end{itemize}
Thus, we are able to remove exactly one edge from a triangle and we infer the correct skeleton, shown in Figure~\ref{fig:cycle}(c). 
  \begin{figure}[h!]
	\centering
	\begin{tabular}{ccc}
	\includegraphics[width=0.25\columnwidth]{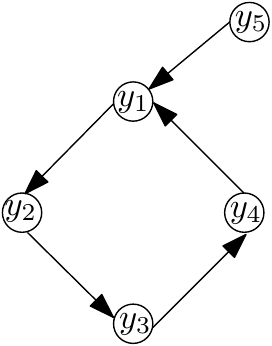} & 
	\includegraphics[width=0.25\columnwidth]{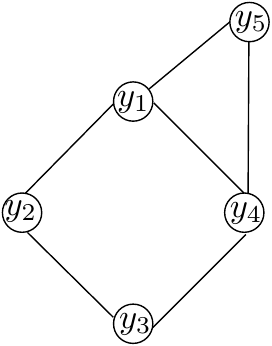} & 
	\includegraphics[width=0.25\columnwidth]{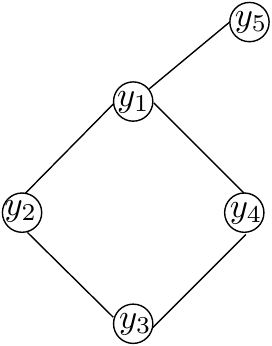} \\
	(a) & (b) & (c)
	\end{tabular}\caption{The causal graph (a) of the LDIM from Example~\ref{ex:4}, its moral graph (b) and its skeleton (c). \label{fig:cycle}}
\end{figure}  
\end{example}

\section{Conclusions}\label{sec:conclusion}

In this work we have studied the problem of topology reconstruction for a network of dynamic systems in presence of direct feedthroughs using observational data.
We have shown that the problem is generally ill-posed, and thus some additional assumptions are needed to be able to uniquely reconstruct the topology.
We have addressed this problem by restricting only the topology of the network without making any assumptions on the dynamics.
For the class of unidirectional triangle-free networks we have described an algorithm capable of either correctly recovering the network topology and providing a certificate of its correctness or recovering an approximated topology with no false positives.
Related algorithms rely on combinatorial searches over the set of nodes. The algorithm presented in this article obtains some information about the moral graph of the network in order to limit the computational complexity of such searches. This strategy is particularly effective when the network has a sparse structure.

\bibliography{trees}

\end{document}